%% file: Ayan_ICCPS19.tex
\documentclass[sigconf, dvipsnames, 9pt]{acmart}

\usepackage{booktabs}
\usepackage{xcolor}

\usepackage{graphicx}
\usepackage{subcaption}
\usepackage{amssymb}
\usepackage{tikz}
\usepackage{amsmath}               
{
	\theoremstyle{plain}
	\newtheorem{assumption}{Assumption}
	\newtheorem{lemma}{Lemma}
	\newtheorem{remark}{Remark}
}

\usepackage{mathtools}

\DeclarePairedDelimiter\floor{\lfloor}{\rfloor}

\usetikzlibrary{arrows, arrows.meta, shapes, shapes.multipart, trees, positioning,
	backgrounds, fit, matrix, petri, calc, automata, chains}
\usepackage{ellipsis}
\usetikzlibrary{calc}
\usetikzlibrary{decorations.pathreplacing,decorations.markings,shapes.geometric}
\usetikzlibrary{patterns}

\tikzset{radiation/.style={{decorate,decoration={expanding waves,angle=90,segment length=4pt}}}}

\definecolor{myred}{RGB}{220,43,25}
\definecolor{mygreen}{RGB}{0,146,64}
\definecolor{myblue}{RGB}{0,143,224}
\definecolor{mygray}{gray}{0.90}
\tikzset{
	myshape/.style={
		rectangle split,
		minimum height=0.2cm,
		rectangle split horizontal,
		rectangle split parts=3, 
		draw, 
		anchor=center,
	}
	
	queuei/.pic={
		\draw[line width=1pt]
		(0,0) -- ++(2cm,0) -- ++(0,-1cm) -- ++(-2cm,0);
		\foreach \Val in {1,...,3}
		\draw ([xshift=-\Val*10pt]2cm,0) -- ++(0,-1cm);
		\node[above] at (1cm,0) {Queue $#1$ $w_{#1}$};   
	},
	mytri/.style={
		draw,
		shape=isosceles triangle,
		isosceles triangle apex angle=60,
		inner xsep=0.65cm
	}
}

\newif\ifcomments

\commentstrue

\newcommand{\setOfULResources}{\mathcal{R}^{\text{UL}}}
\newcommand{\setOfDLResources}{\mathcal{R}^{\text{DL}}}
\newcommand{\numOfULResources}{\mathsf{R^{\text{UL}}}}
\newcommand{\numOfDLResources}{\mathsf{R^{\text{DL}}}}
\newcommand{\ulResourceConsumption}{r_i^{\text{UL}}}
\newcommand{\dlResourceConsumption}{r_i^{\text{DL}}}

\newcommand{\numLoops}{N}
\newcommand{\bs}{\text{BS}}
\newcommand{\avgAge}{\overline{\Delta}}
\newcommand{\iae}{\Sigma_{e}}
\newcommand{\samplingPeriod}{T^s_i}
\newcommand{\E}{\mathbb{E}}
\newcommand{\tr}{\mathsf{tr}}
\newcommand{\ki}{k_i}
\newcommand{\sensor}{\mathcal{S}_i}
\newcommand{\controller}{\mathcal{C}_i}

\newcommand{\errornorm}{\E\left[\left\Vert e_i[\ki] \right\Vert^2 \right]}

 \graphicspath{{./plots/}}

\setcopyright{acmcopyright}
\begin{document}
\copyrightyear{2019}
\acmYear{2019}
\setcopyright{acmcopyright}
\acmConference[ICCPS '19]{10th ACM/IEEE International Conference on Cyber-Physical Systems (with CPS-IoT Week 2019)}{April 16--18, 2019}{Montreal, QC, Canada}
\acmBooktitle{10th ACM/IEEE International Conference on Cyber-Physical Systems (with CPS-IoT Week 2019) (ICCPS '19), April 16--18, 2019, Montreal, QC, Canada}
\acmPrice{15.00}
\acmDOI{10.1145/3302509.3311050}
\acmISBN{978-1-4503-6285-6/19/04}

\title{Age-of-Information vs. Value-of-Information Scheduling for \\ Cellular Networked Control Systems}

\author{Onur Ayan}
\orcid{1234-5678-9012}
\affiliation{%
  \institution{Chair of Communication Networks}
  \institution{Technical University of Munich}
}
\email{onur.ayan@tum.de}

\author{Mikhail Vilgelm}
\orcid{1234-5678-9012}
\affiliation{%
	\institution{Chair of Communication Networks}
	\institution{Technical University of Munich}
}
\email{mikhail.vilgelm@tum.de}

\author{Markus Kl\"ugel}
\orcid{1234-5678-9012}
\affiliation{%
	\institution{Chair of Communication Networks}
	\institution{Technical University of Munich}
}
\email{markus.kluegel@tum.de}

\author{Sandra Hirche}
\orcid{1234-5678-9012}
\affiliation{%
	\institution{Chair of Information-oriented Control}
	\institution{Technical University of Munich}
}
\email{hirche@tum.de}

\author{Wolfgang Kellerer}
\orcid{1234-5678-9012}
\affiliation{%
	\institution{Chair of Communication Networks}
	\institution{Technical University of Munich}
}
\email{wolfgang.kellerer@tum.de}

\renewcommand{\shortauthors}{O. Ayan et al.}

\begin{abstract}
Age-of-Information (AoI) is a recently introduced metric for network operation with sensor applications which quantifies the freshness of data. In the context of networked control systems (NCSs), we compare the worth of the AoI metric with the value-of-information (VoI) metric, which is related to the uncertainty reduction in stochastic processes. First, we show that the uncertainty propagates non-linearly over time depending on system dynamics. Next, we define the value of a new update of the process of interest as a function of AoI and system parameters of the NCSs. We use the aggregated update value as a utility for the centralized scheduling problem in a cellular NCS composed of multiple heterogeneous control loops. By conducting a simulative analysis, we show that prioritizing transmissions with higher VoI improves performance of the NCSs compared with providing fair data freshness to all sub-systems equally.
\end{abstract}



\maketitle

\input{intro.tex}

\input{system.tex}
\input{scheduling.tex}
\input{results.tex}
\input{relatedwork.tex}
\input{conclusions.tex}
\input{ack.tex}

\bibliographystyle{ACM-Reference-Format}

\appendix
\input{appendix.tex}
\end{document}

%% file: intro.tex
\section{Introduction}
Industrial applications form a major driving use case for 5G wireless research. Connectivity within industrial facilities is expected to enable a multitude of novel applications, including remote monitoring, control, and tele-robotics. Most considered scenarios fall into the framework of \textit{networked control systems} (NCSs), where an underlying control loop is closed over a communication medium. Due to the different performance metrics of NCSs compared with traditional network systems, the networking policies need to be adapted not to degrade performance. The wireless communication medium is particularly constrained in spectrum and prone to interference effects, which motivates the problem of prioritization and efficient scheduling of NCSs.

5G cellular networks are envisioned to support machine-type communications (MTC) or machine-to-machine communications (M2M) \cite{Shariatmadari2015}. They refer to a wide spectrum of applications where data communications occur between two or more mobile devices. Process automation, energy grids, healthcare and smart houses are some prominent use cases of M2M / MTC in 5G cellular networks. It is obvious that each of these applications requires different treatment from the communication system point of view due to its distinct features and requirements. Thus, tailoring the communication solutions to the underlying MTC applications can lead to more efficient and reliable services.

Scheduling for NCSs has raised significant interest from a control perspective, where it has been related to time-triggered and event-triggered control. Here, commonly constraints on the available resources (e.g., data rates) are considered in expectation and optimization metrics target the steady-state behavior of an NCS~\cite{5400528,6882817}. While providing optimal stationary policies under certain assumptions, network behavior is often assumed control-agnostic and is abstracted. However, the varying nature of wireless channels, the trade-offs among different control loops, and the coexistence of multiple traffic types in a network in general suggest that gains can be achieved by considering control metrics for network design. In NCS scenarios, it has been shown beneficial to use additional cross-layer metrics for scheduling \cite{mamduhi2017error, vasconcelos2017optimal, hsu2017age}. In particular, two performance metrics raise our interest. Age-of-Information (AoI) is a recently introduced metric for network operation with sensor applications \cite{kaul2012real}. It is a measure of information freshness from the application layer perspective and is applicable for any NCS scenario where there is an uncertainty in the information of interest such as industrial automation or a smart building. Value-of-Information (VoI) quantifies the amount of reduction in the uncertainty of a stochastic process at the recipient. It stems from information theory, originated by Claude E. Shannon in the late 1940s \cite{shannon1948}. While the VoI deals with the content of a new update independently of its timeliness, AoI considers only the timeliness independent of its content. Therefore, age may not be a standalone metric when it comes to monitoring and control of heterogeneous applications sharing the same network. Hence, comparing age and value, we ask the question which of these is more suitable to use in an NCS context.
  
\subsection{Contributions and Outline}
In this paper, we investigate the worth of the AoI and VoI metric for NCSs. We consider a scenario where multiple heterogeneous stochastic control systems are closed over a resource constrained two-hop communication network. Medium access is coordinated by a centralized scheduler that determines which subset of loops are allowed to communicate their up-to-date state information. The deviation of the real state from the augmented state on the receiver, i.e., controller, is considered as performance metric that is also related to the uncertainty in control. 

In this set-up, we are able to show that VoI can be interpreted as a function of the AoI. By designing one scheduler for AoI and one for VoI and conducting a simulative analysis, we show that prioritizing more valuable information leads to lower uncertainty, thus better control performance, than keeping the information at the recipient fresh.

The remainder of the paper is organized as follows. In Section~\ref{sec:scenario}, we introduce the considered scenario and present models for networking and control. Next, we define AoI and VoI in terms of system parameters. Section \ref{sec:Scheduling} presents two scheduling algorithms employing AoI and VoI of each loop as a decision metric. In Section \ref{sec:results}, we illustrate and discuss the results of our simulative study. Section \ref{sec:relatedwork} reviews the related work and Section \ref{sec:conclusions} concludes the paper.

\subsection{Notations}
Throughout this paper $v^T$ and $M^T$ stand for the transpose of a vector $v$ and a matrix $M$, respectively. $\tr(.)$ is the trace operator. The expected value of a random variable $X$ is denoted by $\E\left[ X \right] $. $\left\Vert v \right\Vert$ indicates the euclidean norm of vector $v$ with $\left\Vert v \right\Vert = \sqrt{v^T v}$. The normal distribution with mean $\mu$ and standard deviation $\sigma$ is denoted by $\mathcal{N}(\mu, \sigma^2)$. Additionally, $U(a,b)$ is the uniform distribution with minimum and maximum values $a$ and $b$.

%% file: system.tex
\begin{figure} [t]
	\centering
	\resizebox {\linewidth} {!}
	{\input{scenario.tex}}
	\caption{Scenario: Cellular networked control system with $N$ sub-systems. $\bs$ receives the data from sensors via uplink (UL), and forwards it to the respective controllers via downlink (DL). The scheduler for both UL and DL hops is centralized and located at the $\bs$.}
	\label{fig:scenario}
\end{figure}
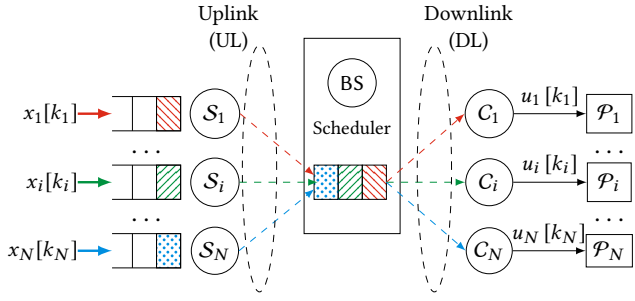

\section{Scenario and Problem Statement}
\label{sec:scenario}
Consider a networked control system shared by $N$ independent, linear time-invariant (LTI) control sub-systems with periodic sampling (see Figure~\ref{fig:scenario}). Each individual sub-system $i$ consists of a plant $\mathcal{P}_i$, a sensor $\mathcal{S}_i$, and a controller $\mathcal{C}_i$. We assume each controller-plant pair to be co-located and hence connected through an ideal controller-to-plant link while the sensor is operating remotely. This is a typical scenario for applications like industrial networked robotics, smart grids or automated highways systems \cite{kheirkhah2015networked, seiler2001analysis}, where the controller observes the plant via remotely deployed sensors or cameras.

\subsection{Network Model}
We assume a cellular network in which every sensor $\sensor$ and controller $\mathcal{C}_i$ are connected to the same base station (BS). Every $\sensor$ transmits observed state information in form of packets in the uplink (UL) direction towards the BS, from which it is forwarded in the downlink (DL) direction to the corresponding controller $\controller$, as shown in Figure \ref{fig:scenario}. The smallest time unit in the system is a transmission slot of unit length, which is indexed by $t \in \mathbb{N}$ in the following. 

Information packets are generated periodically at each sensor, which stores the \textit{latest generated packet} until it is allowed to transmit. If a newer packet is generated while the previous has not yet been transmitted, the sensor replaces older packet with the newer one~\cite{costa2014age}. A centralized dynamic scheduler located at the BS, schedules transmissions on UL, stores the received packets and forwards them on the DL. Again, the BS drops outdated packets and replaces them with newer ones, if received. The scheduling decision vectors on the UL and DL for each time slot are denoted by $\pi^{\text{UL}}(t), \, \pi^{\text{DL}}(t) \in \{0,1\}^N$, where a value of $\pi_i^{\text{UL/DL}}(t) = 1$ indicates that a packet of sub-system $i$ is transmitted over the respective link. We assume that when scheduled, transmissions are received without packet loss at the end of the transmission slot.

As illustrated in Figure \ref{fig:Resource Grid}, uplink and downlink transmissions take place within a time-frequency resource grid. The sets of uplink and downlink resources, $\setOfULResources$ and $\setOfULResources$ are separated in a Frequency-Division-Duplexing (FDD) manner. Formally, $\setOfULResources \cap \setOfDLResources = \emptyset$,  and finite, i.e., $\left|\setOfULResources\right| = \numOfULResources < \infty$, $\left|\setOfDLResources\right| = \numOfDLResources < \infty$. Therefore, the maximum number of simultaneous uplink and downlink transmissions is limited. Throughout the paper, it is assumed that each UL and DL transmission consumes one resource.

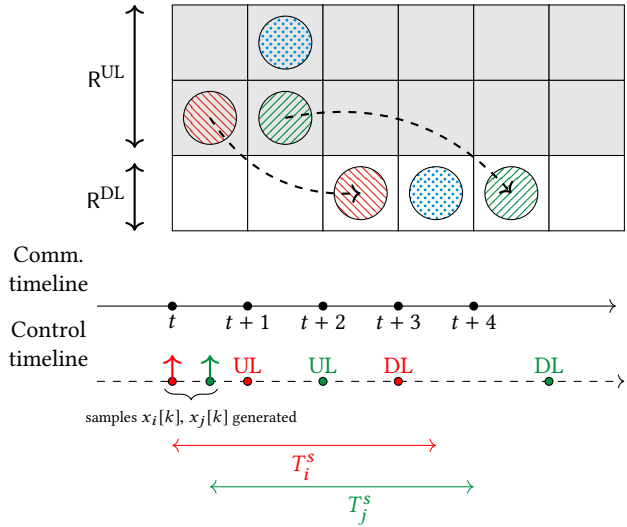
\begin{figure} [t]
	\centering
	\resizebox {\linewidth} {!}
    {\input{resource_grid.tex}}
	\caption{Illustration of the resource grid with communication and control timelines, and with exemplary procedure for sub-system $i$ (red) and $j$ (green). Packet generated by sub-system $i$ is received by the BS at time $(t+1)$, and received by its controller at time $(t+3)$. The packet arrives within the sampling period of the sub-system $T_i^s$, therefore, it is not delayed from the perspective of the sub-system (case 1 in Eqn.~\eqref{eq:Cestimation}). On the contrary, the packet of sub-system $j$ is experiencing delay larger that $T_j^s$, therefore, the packet is considered delayed (case 2 in Eqn.~\eqref{eq:Cestimation}).}
	\label{fig:Resource Grid}
\end{figure}

\subsection{Control Model}
We consider the behavior of the $i$-th sub-system is represented by the following LTI model in discrete time:
\begin{equation}
\label{eq:statespace}
x_i[\ki + 1] = A_i x_i[\ki] + B_i u_i[\ki] + w_i[\ki] 
\end{equation}
with time-step $\ki$, system state $x_i \in \mathbb{R}^{n_i}$, state matrix $A_i \in \mathbb{R}^{n_i \times n_i}$, input matrix $B_i \in \mathbb{R}^{n_i \times m_i}$ and control input $u_i \in \mathbb{R}^{m_i}$. The noise sequence $w_i \in \mathbb{R}^{n_i}$ is considered to be independent and identically distributed (i.i.d) according to a zero-mean Gaussian distribution with diagonal covariance matrix $W_i$. The system state $x_i[\ki]$ with $x_i[0] = w_i[0]$ is measurable by $\sensor$. Each sub-system $i$ generates packets periodically every $T_i^s$ transmission slots with $T_i^s \in \mathbb{N}^+$, where the initial generation happens at slot $T_i^o\sim U(0, \samplingPeriod)$, which is a uniformly distributed random variable. As a consequence, the sub-systems may operate in a non-synchronized fashion, as well as with different update rates. However, we assume that they do not operate faster than the network. The mapping of transmission slots $t$ to sub-system steps $\ki$ becomes:
\begin{equation}
\label{eq:tkmapping}
	\ki(t) = \floor*{\dfrac{t - T_i^o}{\samplingPeriod}}.
\end{equation}

We introduce a variable $\delta_i[\ki] \in \{0, 1\}$ as an indicator of packet reception by the controller. I.e., $\delta_i[\ki] = 1$ if $x_i[\ki]$ is received by $\mathcal{C}_i$ and $\delta_i[\ki] = 0$ if $x_i[\ki]$ is dropped or still waits for transmission at the sensor or BS. The state of a sub-system, as observed by the controller $\mathcal{C}_i$, is given by: 

\begin{equation}
\label{eq:observedState}
z_i[\ki]=
\begin{cases}
x_i[\ki] & \text{, if } \delta_i[\ki] = 1\\
\emptyset & \text{, if } \delta_i[\ki] = 0.
\end{cases}
\end{equation}
Note that due to resource constraints on the downlink, observation of state $x_i[\ki]$ can occur as recent as multiple sampling periods after its generation. Thus, $\controller$ knows the current state of the process only if $\delta_i[\ki(t)] = 1$. 

In order to compensate for packet drops or delays caused by the network, each controller $\controller$ employs a Kalman-like state estimator as in \cite{li2016wireless, sinopoli2004time}. The state estimation is based on the following assumptions:

\begin{assumption}\label{as:1}
	The controller $\controller$ is aware of the system parameters $A_i$, $B_i$ and $W_i$. 
\end{assumption}

\begin{assumption}$T_i^o$ and $\samplingPeriod$ and $t$ are known by the controller.\label{as:2}
\end{assumption}
Assumption \ref{as:1} is motivated by the time-invariant nature of the sub-systems' dynamics. Combined with periodic arrival of samples, Assumption \ref{as:2} implies that $\controller$ is able to map any $t$ to $\ki$ by using Eqn.~\eqref{eq:tkmapping}. Hence, the estimated state $\hat{x}[\ki]$ on the controller side is:
\begin{equation}
\hat{x}_i[\ki] = \E\left[ x_i[\ki]~\big|~\mathcal{I}_i[\ki] \right]
\end{equation}
with the information set $\mathcal{I}_{i}[k]$ available at $\controller$ as follows:
\begin{equation}
\mathcal{I}_{i}[\ki] = \{ \ki, \, z_i[0],\, \dots ,\, z_i[\ki], \, u_i[0], \, \dots , \, u_i[\ki - 1]\}
\end{equation}
Since we are dealing with LTI systems, we assume a stationary control law for each loop: 
\begin{equation}
u_i[k_i] = - L_i \hat{x}_i[k_i]
\end{equation}
where $L_i \in \mathbb{R}^{m_i \times n_i}$ is the state-feedback gain matrix. The scheduler is assumed to be control-aware based on the following assumptions:
\begin{assumption}
	\label{as:3}
	The scheduler at the $\bs$ observes the content of any packet it receives on the UL.
\end{assumption}
\begin{assumption}
	\label{as:4}
	The scheduler is aware of system parameters $A_i$, $W_i$, $B_i$, $L_i$, $T_i^s$, $T_i^o$  $\forall i$.
\end{assumption}
Assumptions \ref{as:3} and \ref{as:4} together enable the scheduler to retain an information set $\mathcal{I}_i^B[k_i]$ as:
\begin{equation}
\mathcal{I}^B_{i}[\ki] = \lbrace \ki, \, z_i^B[0],\, \dots ,\, z_i^B[\ki], \, u_i[0], \, \dots , \, u_i[\ki - 1]\rbrace
\end{equation}
with $z_i^B[\ki]$ depending on a reception variable $\delta_i^B[\ki]$ defined analog to $z_i[\ki]$ and $\delta_i[\ki]$. Because the BS receives data before the controller does, $\delta_i^B[\ki]\geq\delta_i[\ki]$, leading to $\mathcal{I}^B_{i}[k_i] \supseteq \mathcal{I}_i[\ki]$ $\forall i, \ki$. The estimation at the $\bs$ follows analog to that at the controller as:
\begin{equation}
\hat{x}^B_i[\ki] = \E\left[ x_i[\ki]~\big|~\mathcal{I}^B_i[\ki] \right].
\end{equation}
\subsection{Age-of-Information}
If we denote the most recent received observation by $z_i[s_i]$, with $s_i[\ki] = \sup \{ s \in \mathbb{N}: s \leq \ki, z_i[s] \not = \emptyset \}$ the latest control step from which a state has been received, the Age of Information $\Delta_i[\ki]$ at the controller $\controller$ follows as:
\begin{equation}
	\label{eq:AoIDynamics}
	\Delta_i[\ki]= \ki - s_i[\ki]
\end{equation} 
As can be seen, the AoI denotes the number of elapsed control steps since the acquisition of the latest received system state. In contrast to existing literature on AoI, in the given case  $\Delta_i[\ki]$ does not increase linearly with $t$ due to the step-wise mapping of $t$ to $\ki$ given in \eqref{eq:tkmapping}. On the other hand, we argue that AoI evolves linearly with respect to $\ki$ from control perspective. In any case, after each successful DL reception, $s_i$ is increased and the information set $I_i[\ki]$ is extended by $z_i[s_i]$. 

Similarly, if we denote the most recent non-empty observation by $z_i^B[m_i]$ with $m_i[\ki] = \sup \{ m \in \mathbb{N}: m \leq \ki, z_i^B[m] \not = \emptyset \}$, the age of the set $\mathcal{I}_i^B[\ki]$, that is $\Delta_i^B[\ki]$, follows as:
\begin{equation}
\label{eq:AoIBS}
\Delta_i^B[\ki]= \ki - m_i[\ki]
\end{equation} 
It is important to emphasize that $\Delta_i^B[\ki] \leq \Delta_i[\ki]$. In case of equality, i.e., $\mathcal{I}_i[\ki]= \mathcal{I}_i^B[\ki]$, then $\Delta_i^B[\ki] = \Delta_i[\ki]$ holds. To avoid visual clutter, we avoid defining further equations twice both for the $\bs$ and $\controller$. The superscript $(\cdot)^B$ assumes an analogue definition for the $\bs$ of a new introduced variable. In other words, one has to replace $\Delta_i$, $\hat{x}_i$, $z_i$, with $\Delta^B$, $\hat{x}_i^B$, $z_i^B$, respectively.

\subsection{Value-of-Information}
Because AoI is a variable defined in the units of control steps, we can use it to express control variables. 
\begin{lemma}
	\label{lem:Lemma1}
	Given the information set $\mathcal{I}_{i}[\ki]$ and the age-of-information $\Delta_i[\ki]$, the estimated plant state is determined by:
	\begin{equation}
	\label{eq:Cestimation}
	\hat{x}_i[\ki]=
	\begin{cases}
	x_i[\ki] & \text{, if } \Delta_i[\ki] = 0\\
	f(\Delta_i[\ki],~\mathcal{I}_i[\ki]) & \text{, if } \Delta_i[\ki] > 0
	\end{cases}
	\end{equation}
	with:
	\begin{equation}
	\label{eq:fFunction}
	f(\Delta_i[\ki],~\mathcal{I}_i[\ki]) \triangleq A_i^{\Delta_i[\ki]} z_i[s_i] + \sum_{q=1}^{\Delta_i[\ki]}A_i^{q-1} B_i u_i[\ki-q]
	\end{equation}
\end{lemma}
\begin{proof}
	The proof is given in Appendix \ref{app:proof_lemma_1}.
\end{proof}
If $\Delta_i[\ki]$ is zero, it means that the controller has been provided the latest plant state. Otherwise, the current state estimate $\hat{x}_i[k]$ is recursively calculated from the most recent information received by the controller which is $z_i[s_i]$ as stated above. Thus, the estimation error induced by the network is defined as the difference between the true and estimated states as:
\begin{equation}
e_i[\ki] = x_i[\ki] - \hat{x}_i [\ki] \nonumber = \sum_{q=1}^{\Delta_i[\ki]}A_i^{q-1} w_i[\ki-q] 
\label{eq:nie}
\end{equation}

\begin{lemma}
	\label{lem:Lemma2}
	Given the age-of-information $\Delta_i[\ki]$, noise covariance matrix $W_i$, and system matrix $A_i$, the quadratic error norm can be estimated as follows:
	\begin{equation}
	\errornorm =
	\begin{cases}
	0 & \text{, if } \Delta_i[\ki] = 0\\
	g\left(\Delta_i[\ki]\right) & \text{, if } \Delta_i[\ki] > 0
	\end{cases},
	\label{eq:quadraticerrorcases}
	\end{equation}
	with:
	\begin{equation}
	\label{eq:quadraticerrornormestimation}
	g(\Delta_i[\ki]) \triangleq
	\sum_{r=0}^{\Delta_i[\ki]-1} \tr\left((A_i^T)^r A_i^r W_i\right).
	\end{equation} 
\end{lemma}
\begin{proof}
	The proof is given in Appendix \ref{app:proof_lemma_2}.
\end{proof}
Note that $g:\mathbb{N} \to \mathbb{R}$ is strictly increasing for any invertible $A_i$ and positive-definite noise covariance matrix $W_i$. Analogously, for $\Delta_i^B[\ki] > 0 $, we define the $\bs$ counterparts of $\hat{x}_i[\ki]$, $e_i[\ki]$ and $\E\left[\left\Vert e_i[\ki]\right\Vert^2 \right]$ as follows:
\begin{align}
&e_i^B[\ki] = x_i[\ki] - \hat{x}^B_i[\ki] \\
&\hat{x}^B_i[\ki] = f(\Delta_i^B[\ki],~\mathcal{I}^B_i[\ki]) \\
&\E\left[\left\Vert e_i^B[\ki]\right\Vert^2 \right] = g(\Delta_i^B[\ki]).
\end{align} 

\begin{figure}[t]
	\centering
	\includegraphics[width=\linewidth]{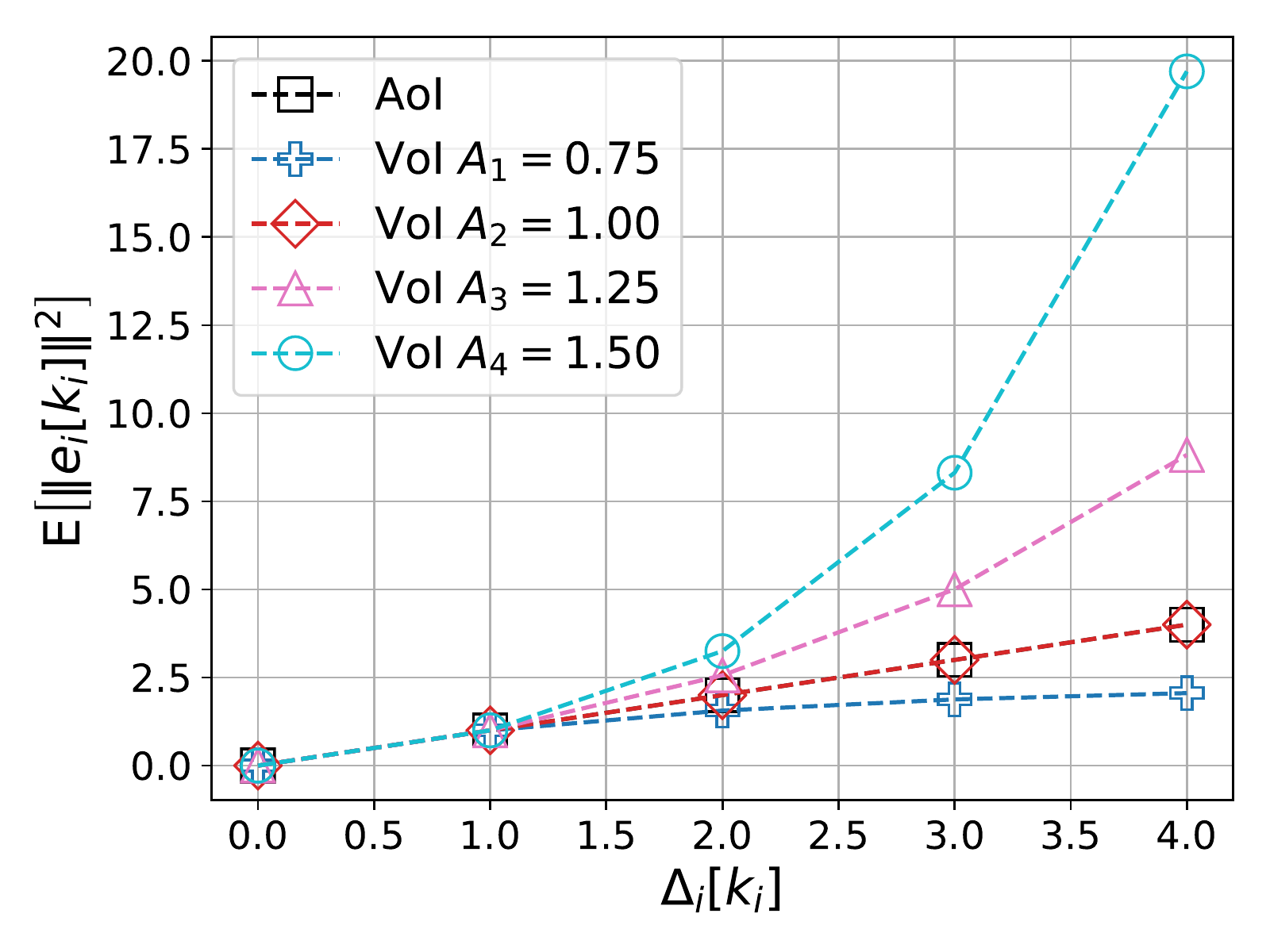}
	\caption{VoI defined as expectation of a quadratic estimation error norm, a function of AoI for an individual sub-system $i$ with different types of plants dynamic $A_i\in\{0.75,1,1.25,1.50\}$ (scalar system assumed for illustrative purposes).}
	\label{fig:voiNonHD}
\end{figure}

Figure \ref{fig:voiNonHD} shows the behavior of expected quadratic error norm, i.e., $\errornorm$ , as a functional of age. To that end, we have selected 4 type of scalar plants with $A_1 = 0.75$, $A_2 = 1.0$, $A_3 = 1.25$ and $A_4$ = 1.5 and kept the noise covariance matrix constant at $W = 1$ for all of them.  The black curve labeled with AoI corresponds the $\Delta_i[\ki]$ which is a line with slope 1. An interesting aspect is that for $A_1 = 0.75$, which is a stable system, the error converges to a finite value for infinitely large AoI. It can easily be shown by letting $\Delta_i \rightarrow \infty$ and applying convergence condition of power series on $g\left( \Delta_i \right)$ from Eqn.~(\ref{eq:quadraticerrornormestimation}).

We propose a link-based value-of-information metric, i.e., $v_i^{\text{UL/DL}}$ both for the UL and the DL. VoI is defined as a measure of uncertainty reduction from the information set of the receiver in case of a successful transmission. In case of an uplink packet, VoI is defined as:
\begin{align}
v_i^{\text{UL}}(t) &= \mathbb{E}\left[\left\Vert e_i^B[\ki] - e_i^S[\ki]\right\Vert^2 \right] \nonumber \\
&=  \mathbb{E}\left[\left\Vert e_i^B[\ki]\right\Vert^2 \right] 
\end{align}
with $\ki = \ki(t)$ as in Eqn. (\refeq{eq:tkmapping}). $e_i^S[\ki]$ is defined as the measurement error on the sensor side which is assumed to be zero throughout this paper, i.e., $e_i^S[\ki] = 0$. Similarly for a DL packet:
\begin{align}
v_i^{\text{DL}}(t) &= \E\left[\left\Vert e_i[\ki] - e_i^B[\ki]\right\Vert^2 \right] \nonumber \\
&=  \left\Vert\hat{x}^B_i[\ki] - \hat{x}_i[\ki] \right\Vert^2
\end{align}
Note that, $v_i^{DL}(t)$ does not contain any non-determinism due to Assumption \ref{as:3} and Assumption \ref{as:4}.

%% file: scenario.tex
	\begin{tikzpicture}[>=latex]
	
	\coordinate (q1) at (0, -0.25); 
	\draw[pattern=north west lines, pattern color=myred] (1cm,0) rectangle (1cm -10pt, -0.5);
	\draw (0,0) -- ++(1cm,0) -- ++(0,-0.5cm) -- ++(-1cm,0);
	\foreach \i in {1,...,2}
	\draw (1cm-\i*10pt,0) -- +(0,-0.5cm);
	
	\node[align=center] at (0.5, -0.8) {\Large $\mathbf{\dots}$};	
	
	\coordinate (q2) at (0, -1.25); 
	\draw[pattern=north east lines, pattern color=mygreen] (1cm,-1) rectangle (1cm -10pt, -1.5);
	\draw (0,-1) -- ++(1cm,0) -- ++(0,-0.5cm) -- ++(-1cm,0);
	\foreach \i in {1,...,2}
	\draw (1cm-\i*10pt,-1) -- +(0,-0.5cm);
	
	\node[align=center] at (0.5, -1.8) {\Large $\mathbf{\dots}$};	
	
	\coordinate (q3) at (0, -2.25); 
	\draw[pattern=crosshatch dots, pattern color=myblue] (1cm,-2) rectangle (1cm -10pt, -2.5);
	\draw (0,-2) -- ++(1cm,0) -- ++(0,-0.5cm) -- ++(-1cm,0);
	\foreach \i in {1,...,2}
	\draw (1cm-\i*10pt,-2) -- +(0,-0.5cm);
	
	\node[draw,circle,inner sep=0.1cm] at (1.5, -0.25)
	(sensor1) 
	{$\mathcal{S}_1$};
	
	\node[draw,circle,inner sep=0.1cm] at (1.5, -1.25)
	(sensor2) 
	{$\mathcal{S}_i$};
	
	\node[draw,circle, inner sep=0.055cm] at (1.5, -2.25)
	(sensori) 
	{$\mathcal{S}_N$};
	
	\node[draw, circle, align=center] at (3.5, 0.2)
	(node1) 
	{$\mathsf{BS}$};
	
	\node[align=center, inner sep=0.1cm] at (3.5, -.45)
	(sched) 
	{\small Scheduler};
	\draw[align=center] (2.8, -2) rectangle (4.2,0.85);

	\draw[pattern=north west lines, pattern color=myred] (4cm,-1) rectangle (4cm -10pt, -1.5);
	
	\draw[pattern=north east lines, pattern color=mygreen] (4cm-10pt,-1) rectangle (4cm -20pt, -1.5);
	
	\draw[pattern=crosshatch dots, pattern color=myblue] (4cm-20pt,-1) rectangle (4cm -30pt, -1.5);
	
	\node[draw,circle, inner sep=0.1cm] at (5.5, -0.25)
	(controller1) 
	{$\mathcal{C}_1$};
	\node[draw,circle, inner sep=0.1cm] at (5.5, -1.25)
	(controller2) 
	{$\mathcal{C}_i$};
	\node[draw,circle, inner sep=0.055cm] at (5.5, -2.25)
	(controlleri) 
	{$\mathcal{C}_N$};
	
	\node[draw,rectangle, inner sep=0.15cm] at (7.25, -0.25)
	(plant1) 
	{$\mathcal{P}_1$};
	\node[draw,rectangle, inner sep=0.15cm] at (7.25, -1.25)
	(plant2) 
	{$\mathcal{P}_i$};
	\node[draw,rectangle, inner sep=0.10cm] at (7.25, -2.25)
	(planti) 
	{$\mathcal{P}_N$};
	
	\node[align=center] at (7.25, -0.8) {\Large $\mathbf{\dots}$};		
	\node[align=center] at (7.25, -1.8) {\Large $\mathbf{\dots}$};	
	
	\draw[line width=1.0pt, myred,->] (-0.5, -0.25) -- (q1) 
	node[color =black, near start,left]{$x_1[k_1]$};
	
	\draw[line width=1.0pt, mygreen,->] (-0.5, -1.25) -- (q2) 
	node[color =black, near start, left]{$x_i[k_i]$};
	
	\draw[line width=1.0pt, myblue,->] (-0.5, -2.25) -- (q3) 
	node[color =black, near start, left]{$x_N[k_N]$};
	
	\draw[->,myred, dashed] (sensor1.east) --  (2.95, -1.15);
	\draw[->,mygreen, dashed] (sensor2.east) --  (3, -1.25);
	\draw[->,myblue, dashed] (sensori.east) --  (2.95, -1.35);
	
	\draw[->,myred, dashed] (4.00, -1.25) --  (controller1.west);
	\draw[->,mygreen, dashed] (4, -1.25) --  (controller2.west);
	\draw[->,myblue, dashed] (4.00, -1.25) --  (controlleri.west);
	
	\draw[->,black] (controller1.east) --  (plant1.west) node[midway, above]{$u_1\left[k_1\right]$};
	\draw[->,black] (controller2.east) --  (plant2.west) node[midway, above]{$u_i\left[k_i\right]$};
	\draw[->,black] (controlleri.east) --  (planti.west) node[midway, above]{$u_N\left[k_N\right]$};
	
	\draw[dashed] (2.15,-1.05) ellipse (0.25cm and 1.8cm);
	\draw[dashed] (4.7,-1.05) ellipse (0.25cm and 1.8cm);
	\node[align=center] at (1.7, 1)	{Uplink\\(UL)};
	\node[align=center] at (5.2, 1)	{Downlink\\(DL)};
			
	\end{tikzpicture}

%% file: resource_grid.tex
\begin{tikzpicture}
\foreach \x in {0,...,5} 
	\foreach \y in {0,...,1}
		\filldraw[fill=mygray, draw=black] (\x, \y) rectangle (\x + 1, \y + 1);
		
\foreach \x in {0,...,5} 
	\foreach \y in {-1,...,-1}
		\filldraw[fill=White, draw=black] (\x, \y) rectangle (\x + 1, \y + 1);

%


\draw [<->, thick] (-0.5,0.1) -- node[align=left, left] {$\numOfULResources$} (-.5, 2);
\draw [<->, thick] (-0.5,-0.1) -- node[align=left, left] {$\numOfDLResources$} (-.5, -1);

\draw[pattern=north west lines, pattern color=myred] (0.5,0.5) circle (10pt) node[align=left,   below] (redUL) {} ;
\draw[pattern=north east lines, pattern color=mygreen] (1.5,0.5) circle (10pt) node[align=left,   below] (greenUL) {};
\draw[pattern=crosshatch dots, pattern color=myblue] (1.5,1.5) circle (10pt) node[align=left,   below] {};

\draw[pattern=north west lines, pattern color=myred] (2.5,-0.5) circle (10pt) node[align=left,   below] (redDL) {} ;
\draw[pattern=north east lines, pattern color=mygreen] (4.5,-0.5) circle (10pt) node[align=left,   below] (greenDL) {};
\draw[pattern=crosshatch dots, pattern color=myblue] (3.5,-0.5) circle (10pt) node[align=left,   below] {};

\draw[->, dashed, black, thick] (redUL.north)  to[bend right] (redDL.north);
\draw[->, dashed, black, thick] (greenUL.north)  to[bend left] (greenDL.north);


%
	

\newcommand{\taxisY}{-2}
\node[left, black, align=center] (t) at (-1, \taxisY+0.5){Comm.\\timeline};

\node[](tAxis) at (6, \taxisY) {};
\coordinate[] (torigin) at (-1, \taxisY);
\draw[->] (torigin) -- (tAxis);

\draw[black,fill=black] (0,\taxisY) circle (1.5pt);
\draw[black,fill=black] (1,\taxisY) circle (1.5pt);
\draw[black,fill=black] (2,\taxisY) circle (1.5pt);
\draw[black,fill=black] (3,\taxisY) circle (1.5pt);
\draw[black,fill=black] (4,\taxisY) circle (1.5pt);
\node[below] (t) at (0, \taxisY){$t$};
\node[below] (t-1) at (1, \taxisY){$t+1$};
\node[below] (t-1) at (2, \taxisY){$t+2$};
\node[below] (t-1) at (3, \taxisY){$t+3$};
\node[below] (t) at (4, \taxisY){$t+4$};


\newcommand{\caxisY}{-3}
\draw[->, dashed, black] (-1, \caxisY) -- (6,\caxisY);
\node[left, black, align=center] (t) at (-1, \caxisY+0.5){Control\\timeline};

\draw[<->, red] (0, \caxisY-0.85) -- (3.5,\caxisY-0.85)
node[below, red, align=center] (t) at (1.75, \caxisY-0.85){$T_i^{s}$};

\draw[black,fill=red] (0,\caxisY) circle (1.5pt);
\draw[->, thick, red] (0, \caxisY) -- (0,\caxisY+.35);

\draw[black,fill=red] (1,\caxisY) circle (1.5pt);
\node[above, red, align=center] (t) at (1, \caxisY){UL};

\draw[black,fill=red] (3,\caxisY) circle (1.5pt);
\node[above, red, align=center] (t) at (3, \caxisY){DL};


\draw [decorate,decoration={brace,amplitude=4pt,mirror}] (-0.1,\caxisY-0.1) -- (0.6,\caxisY-0.1) 
node [black,midway, yshift=-0.4cm,align=center] {\scriptsize samples $x_i[k],x_j[k]$ generated};

\draw[<->, mygreen] (0.5, \caxisY-1.4) -- (4,\caxisY-1.4)
node[below, mygreen, align=center] (t) at (2.5, \caxisY-1.4){$T_j^{s}$};
\draw[black,fill=mygreen] (0.5,\caxisY) circle (1.5pt);
\draw[->, thick, mygreen] (0.5, \caxisY) -- (0.5,\caxisY+.35);

\draw[black,fill=mygreen] (2,\caxisY) circle (1.5pt);
\node[above, mygreen, align=center] (t) at (2, \caxisY){UL};

\draw[black,fill=mygreen] (5,\caxisY) circle (1.5pt);
\node[above, mygreen, align=center] (t) at (5, \caxisY){DL};

%
%
%
%


%
%

\end{tikzpicture}

%% file: scheduling.tex
\section{Joint Scheduling Design}
\label{sec:Scheduling}
Due to resource constraints on both hops, centralized scheduler at the BS prioritizes sub-systems based on performance metrics. We define two schedulers utilizing AoI and VoI as metrics. They follow a joint design for uplink and downlink.

The fact that we assume equal channel qualities among loops allows us to distinguish between two cases: (i) Uplink is the bottleneck of the network, i.e., $\numOfULResources \leq \numOfDLResources$ and (ii) Downlink is the bottleneck, i.e., $\numOfDLResources \leq \numOfULResources$. In the first case, all uplink transmissions received by $\bs$ are forwarded as soon as the data reception is completed. Hence, the DL/UL scheduling problem can be reduced to a single-hop problem, where $\bs$ and the $\controller$ nodes are logically merged together. In the second case, downlink hop is limiting the network throughput, therefore, joint scheduling problem for both links must be considered.

\begin{remark}
We implicitly assume that every scheduled transmission is successful. To ensure this, cellular networks typically employ re-transmission techniques, e.g., hybrid automatic repeat request. We note that (heterogeneous) packet loss probability can be readily accommodated into the scheduler design by weighting respective AoI or VoI metrics by the expected packet success probability.
\end{remark}

\subsection{Age-of-Information Scheduler}
\label{subsec:AoI}
As the name suggests, AoI scheduler aims to prevents staleness of information sets at the controller side. The targeted problem is formalized as:

\begin{subequations}
\begin{align}
& \min\limits_{\pi^{\text{UL}}(t), \pi^{\text{DL}}(t)}
& &\limsup\limits_{T\to\infty}\frac{1}{T}\sum_{t=0}^{T-1} \sum_{i=1}^{N} \Delta_i(t)&\label{eq:aoiMin}\\
& \text{subject to}
& &\sum_{i=1}^{N} \pi_i^{\text{UL}}(t) \leq \numOfULResources,&\label{eq:aoiMin_ul_constr}\\
& & & \sum_{i=1}^{N} \pi_i^{\text{DL}}(t) \leq \numOfDLResources&\label{eq:aoiMin_dl_constr}
\end{align}
\end{subequations}

To solve the above problem, we can leverage results for single-hop AoI optimization \cite{Kadota2018}, according to which greedy scheduling is in fact age-optimal if all uplink transmissions have the same success probability. As this is true under our assumptions, where all transmissions are always successful, we can make use of the results to extend towards the two-hop case, which we do in the following Lemma:

\begin{lemma}
	Assume that $\numOfULResources < \numOfDLResources$ and that the sequence of uplink schedules $\lbrace \pi^{UL}(1),\pi^{UL}(2),...\rbrace$ is age-optimal for the uplink hop. Then, by creating a sequence of downlink schedules as $\pi^{\text{UL}}(t) = \pi^{\text{UL}}(t-1)$ $\forall t \geq 2$, the combination of uplink and downlink sequence is age-optimal for the two-hop case. Further, assuming  $\numOfULResources \geq \numOfDLResources$ and that the sequence of downlink schedules $\lbrace \pi^{\text{DL}}(2), \, \pi^{\text{DL}}(3), \, ...\rbrace$ is age-optimal for the downlink hop, we can create a sequence of uplink schedules as  $\pi^{\text{UL}}(t) = \pi^{\text{DL}}(t+1)$ such that the combination of both sequences is age-optimal.
\end{lemma}
\begin{proof}
	Consider the case of $\numOfDLResources < \numOfULResources$ and observe that if $\pi^{\text{UL}}(t)$ satisfies \eqref{eq:aoiMin_ul_constr}, it also satisfies \eqref{eq:aoiMin_dl_constr}. Consider the case that $\pi^{\text{DL}}(t)\neq \pi^{\text{UL}}(t-1)$. By replacing it with $\tilde{\pi}^{\text{DL}}(t):=\pi^{\text{UL}}(t-1)$ we achieve that $\delta_i[k_i(t-1)] = 1$ $\forall i:\pi_i^{\text{UL}}(t-1)=1$. Hence, the information set $\tilde{\mathcal{I}}_i[k_i(t)]\supset\mathcal{I}_i[k_i(t)]$, yielding $\tilde{s}_i[k_i(t)]\geq s_i[k_i(t)]$ and $\tilde{\Delta}_i[k_i(t)]\leq\Delta_i[k_i(t)]$, respectively. 
	
	Now assume that $\numOfDLResources \geq \numOfULResources$ and observe that if $\pi^{\text{DL}}(t)$ satisfies \eqref{eq:aoiMin_dl_constr}, it also satisfies \eqref{eq:aoiMin_ul_constr}. Consider the case that $\pi^{\text{UL}}(t)\neq \pi^{\text{DL}}(t+1)$. By replacing it with $\tilde{\pi}^{\text{UL}}(t):= \pi^{\text{DL}}(t+1)$ we achieve that $\delta_i[k_i(t)] = 1$ $\forall i:\pi_i^{\text{DL}}(t+1)=1$. Hence, again the information set $\tilde{\mathcal{I}}_i[k_i(t+1)]\supset\mathcal{I}_i[k_i(t+1)]$, yielding the same result, respectively.
\end{proof}
The intuitive explanation of the Lemma is the following: Assuming that the uplink resources form a bottleneck, anything that has been transmitted on the uplink can directly be forwarded on the downlink. Choosing not to transmit artificially adds an increased AoI that can be avoided. On the other hand, if the downlink resources form a bottleneck, any transmission on the downlink can be matched by fetching the corresponding sensor value one step before. Not doing so again artificially adds an increased AoI. In both cases, it is sufficient to know the optimal decision for only one of the hops, which has been proven to be the greedy solution in~\cite{Kadota2018}.

\subsection{Value-of-Information Scheduler}
\label{subsec:VoI}
We propose an application-aware scheduling algorithm that is jointly allocating resources on both hops. The scheduler obtains the value of each UL and DL packet as a function of age-of-information at each hop and aims to minimize the overall \textit{quadratic network induced error norm} in expectation:
\begin{equation}
\begin{aligned}
& \min\limits_{\pi^{\text{UL}}(t), \pi^{\text{DL}}(t)}
& & \limsup\limits_{T\to\infty}\frac{1}{T}\sum_{t=0}^{T-1} \sum_{i=1}^{N} \E\left[\left\Vert e_i\left[k_i(t)\right]\right\Vert^2 \right]\\
& \text{subject to}
& & \sum_{i=1}^{N} \pi_i^{\text{UL}}(t) \leq \numOfULResources, \\
& & & \sum_{i=1}^{N} \pi_i^{\text{DL}}(t) \leq \numOfDLResources
\end{aligned}
\label{eq:voiMin}
\end{equation}
The scheduling problem in (\ref{eq:voiMin}) is a combinatorial optimization problem and not solvable in polynomial time. By applying dynamic programming, it can be solved for a given finite horizon. However, finding the global optimal solution of (\ref{eq:voiMin}) is computationally very expensive, and is out of scope for this paper as it is not applicable for dynamic schedulers. Instead, we accommodate greedy solution for both links separately where we maximize the transmitted value-of-information at single slot on the uplink as:
\begin{equation}
\begin{aligned}
& \max\limits_{\pi^{\text{UL}}(t)}
& & \sum_{i=1}^{N} \pi_i^{\text{UL}}(t) \cdot v_i^{\text{UL}}(t)\\
& \text{subject to}
& & \sum_{i=1}^{N}\pi_i^{\text{UL}}(t) \leq \numOfULResources, \\
\end{aligned}
\end{equation}
and on the downlink as:
\begin{equation}
\begin{aligned}
& \max\limits_{\pi^{\text{DL}}(t)}
& & \sum_{i=1}^{N} \pi_i^{\text{DL}}(t) \cdot v_i^{\text{DL}}(t)\\
& \text{subject to}
& & \sum_{i=1}^{N}\pi_i^{\text{DL}}(t) \leq \numOfDLResources.\\
\end{aligned}
\end{equation}
Our solution provides an upper bound for the optimal cost function of the problem~\eqref{eq:voiMin}. In Section \ref{sec:results}, we show that even the upper bound by scheduling based on VoI outperforms the optimal AoI scheduling.

%% file: results.tex
\section{Numerical Evaluation}
\label{sec:results}
In this section, we present a simulative analysis and comparison of the schedulers defined in Section \ref{sec:Scheduling}.  

\subsection{Simulation Setup}
We simulate an exemplary set-up with heterogeneous scalar LTI sub-systems, where $\text{cl}=4$ classes have different state matrices $A_i\in \{0.75,1,1.25,1.5\}$. The number of sub-systems $N^{(j)}$ corresponding to a plant class $j$ is assumed to be equal for each $j$, as we vary the total number of sub-systems $N\triangleq\sum_{j=1}^{\text{cl}}N^{\left(j\right)}$. The state-feedback gain matrix is chosen according to deadbeat control strategy $L_i = A_i$ \cite{mamduhi2014event}. Input matrix is equal among loops $B_i = 1,\, \forall i \in \{1,\,\dots,\,N\}$. System noise is given by $w_i \sim \mathcal{N}(0,\,1)$. For the sake of simplicity, we assumed all transmissions to require single time-frequency resource, i.e., $\ulResourceConsumption = \dlResourceConsumption = 1$. We consider equal sampling period for all control loops, i.e., $\samplingPeriod  = 10,\, \forall i \in \{1,\,\dots,\,N\}$. Number of downlink resources is chosen as $\numOfDLResources = 3$ and number of uplink resources is varied between $\numOfULResources \in \{ 1, \, 2, \, 3, \, 6, \, 9  \}$. Simulation run-time $T_{\text{sim}}$ is 20000 transmission slots.

As the performance indicators, we use the average AoI per control loop, i.e., $\avgAge$, to represent information staleness and the Integrated Absolute Error (IAE) per loop, i.e., $\iae$, to quantify the uncertainty in the controlled process. $\avgAge$ and $\iae$ are defined as follows:
\begin{equation}
\avgAge
= \dfrac{1}{\numLoops} \dfrac{1}{T_{\text{sim}}} \sum_{i = 1}^{N} \sum_{t = 0}^{T_{\text{sim}}-1} \Delta_i \left( t \right)
\end{equation}
\begin{equation}
\iae
= \dfrac{1}{\numLoops} \sum_{i = 1}^{N} \sum_{t = 0}^{T_{\text{sim}}-1} \lVert e_i[k_i(t)] \rVert
\end{equation}

\begin{figure}[t]
	\includegraphics[width=\linewidth]{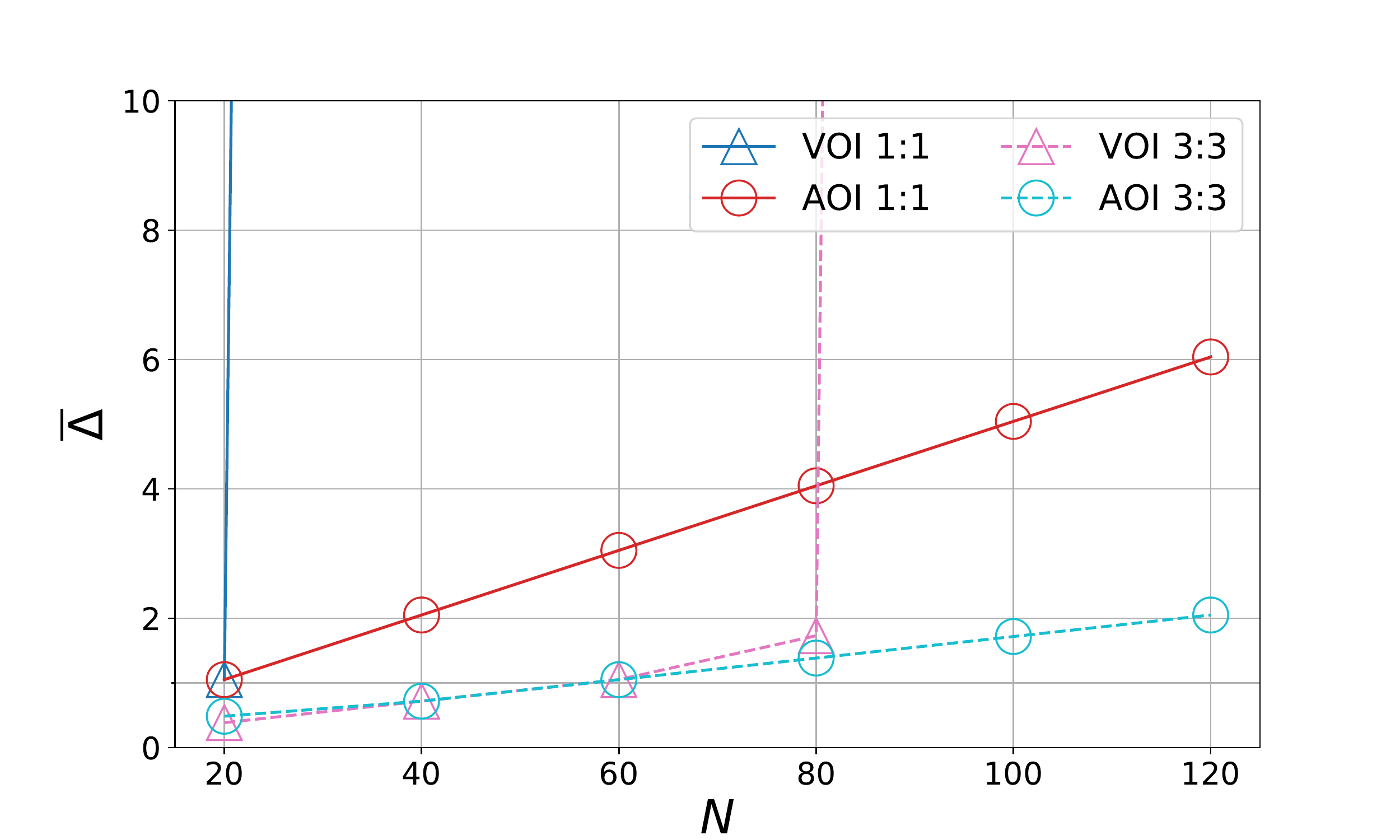}
	\caption{Average AoI per sub-system as a function of the total number of syb-systems $N$, for two configurations of UL/DL resources: $\numOfULResources : \numOfDLResources = \{ 1:1, \, 3:3\}$.}
	\label{fig:AOI}
\end{figure}

\subsection{Results and Evaluation}
First, we investigate the response of the performance metrics to varying number of sub-systems and resources in the network. Figure \ref{fig:AOI} illustrates the average age-of-information per control loop as $\numLoops$ increases for different amounts of available resources $\numOfDLResources=\numOfULResources=3$ and $\numOfDLResources=\numOfULResources=1$. The figure presents results for a neutral DL/UL configuration with equal amounts of resources, where neither hop is a bottleneck. Given $\numLoops = 20$ and $\numOfULResources = \numOfDLResources = 3$, both type of schedulers provide similar performance in terms of $\avgAge$. As $\numLoops$ increases linearly, we observe for the AoI scheduler that the average age per loop is increasing linearly as well. This is expected since AoI scheduler treats all type of plants equally fair and thus information staleness in the network becomes directly proportional to the total amount of resources available in the network. A linear dynamics is also observed for the $\numOfDLResources = \numOfDLResources = 1$ case but with a higher slope, since less network resources are available.

On the other hand, the effect of the unfair treatment of sub-systems by the VoI scheduler becomes evident from the drastic increase of $\avgAge$ after $\numLoops = 20$ and $\numLoops = 80$ for $\numOfULResources = \numOfDLResources = 1$ and $\numOfULResources = \numOfDLResources = 3$ configurations, respectively. This coincides with the average AoI per loop to exceed one, i.e., $\avgAge > 1$ being consistent with Figure \ref{fig:voiNonHD}. From $\numLoops = 40$ on for the $\numOfULResources = \numOfDLResources = 1$ scenario and from $\numLoops = 100$ on for the $\numOfULResources = \numOfDLResources = 3$ scenario, the average AoI $\avgAge$ goes to infinity. This follows from the fact that $\E\left[\left\Vert e_i[\ki] \right\Vert^2 \right]$ converges for plants with $A_i = 0.75$. It can easily be shown by letting $\Delta_i[\ki] \rightarrow \infty$ and applying convergence condition of power series on $g\left( \Delta_i[\ki] \right)$ from Eqn.~(\ref{eq:quadraticerrornormestimation}). As a result of the convergence property, plants $i$ with $A_i=0.75$ never get to transmit as they are dominated by non-converging type of plants with $A_i\geq 1$. 
\begin{figure}[t]
	\includegraphics[width=\linewidth]{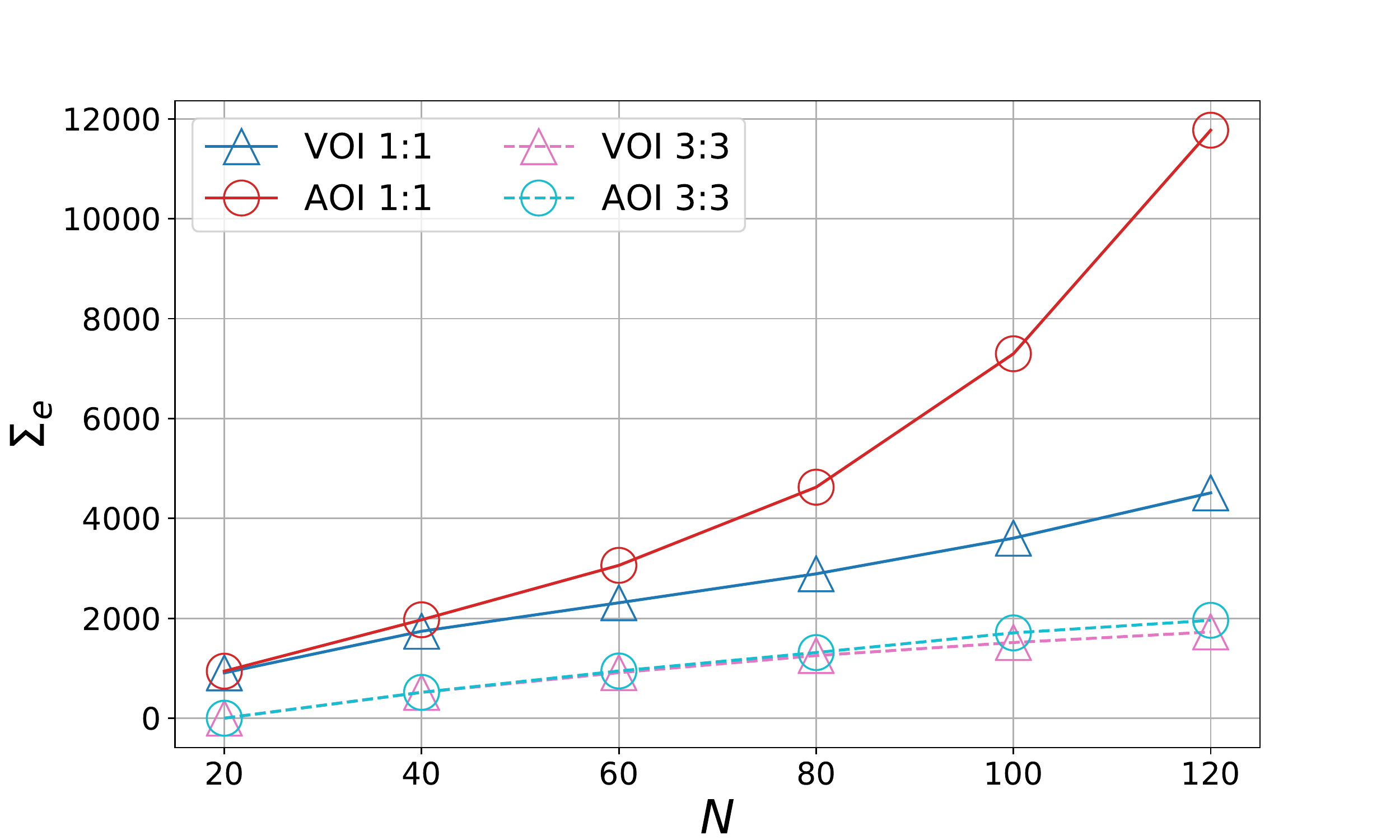}
	\caption{Integrated Absolute Error per sub-system as a function of the total number of sub-systems $N$, for two configurations of UL/DL resources: $\numOfULResources : \numOfDLResources = \{ 1:1, \, 3:3\}$.}
	\label{fig:INIE}
\end{figure}

Now, let us have a look at how $\iae$ is affected by an increase of $\numLoops$. In Figure \ref{fig:INIE} we illustrate how both schedulers perform with respect to reducing the network induced error per loop. From the figure, it is evident that VoI scheduler outperforms the AoI scheduler in $\iae$ metric even though the fairness in age-of-information was not delivered. As we can see, with increasing inadequacy of available resources the gap between the AoI- and VoI scheduler expands faster. This is caused by the non-linear dynamics of network induced error with increasing age-of-information, as visible in Figure \ref{fig:voiNonHD}. Note that, having three uplink and three downlink resources provides triple amount of throughput in average than having one resource in uplink and downlink each. Therefore, in Figure \ref{fig:INIE} the resulting $\avgAge$ and $\iae$ at $N=120$ with $\numOfULResources = \numOfDLResources = 3$ is very close to the $\avgAge$ and $\iae$ values at $N=40$ with $\numOfULResources = \numOfDLResources = 1$. This is also the case for $N=20$ and $N=60$ in Figure \ref{fig:AOI}.
\begin{figure}[t]
	\includegraphics[width=\linewidth]{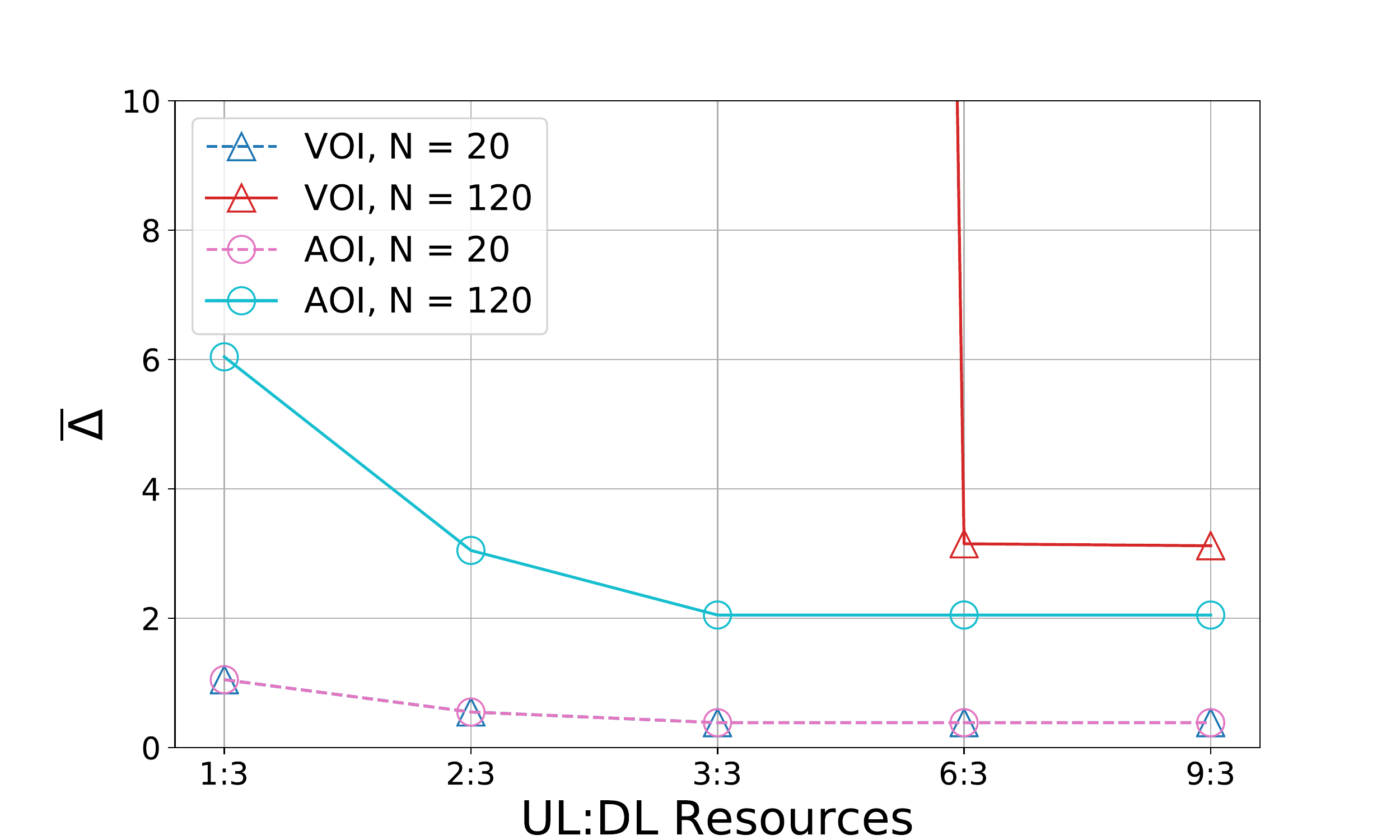}
	\caption{Sensitivity analysis of the average AoI to the UL/DL configuration, with the ratio $\frac{\numOfULResources}{\numOfULResources}$ on the $x$-axis. The number of DL resources is kept fixed $\numOfDLResources=3$, and the number of UL resources is varying $\numOfULResources\in\{1,\dots,9\}$. Left part of the plot ($\numOfULResources<3$) represents an UL bottleneck scenario, while the right part ($\numOfULResources>3$) represents a DL bottleneck scenario.}
	\label{fig:AOI_sens}
\end{figure}

We further investigate the sensitivity of the selected performance indicators to variations in UL/DL resource configuration, by increasing the number of uplink resources $\numOfULResources$ for fixed $\numOfDLResources$. This illustrates a shift of the resource bottleneck from UL to DL. Figures \ref{fig:AOI_sens} and \ref{fig:INIE_sens} show $\avgAge$ and $\iae$ for $\numOfDLResources = 3$ and $\numOfULResources \in \{ 1, \, 2, \, 3, \, 6, \, 9  \}$. We select $N=20$ and $N=120$ as representation of low and high resource demand scenarios, respectively. 

For the low demand case with $\numLoops = 20$, both schedulers produce similar results due to resource abundance in the network. As we cut UL resources down, $\numOfULResources \in \{2, 1\}$, the resulting performance in terms of both indicators decreases due to lower throughput provided. Adding more resources on the uplink, i.e., $\numOfULResources \in \{ 3, 6, 9 \}$ does not have any effect since all sub-systems are provided sufficient transmission opportunities. 

For the high demand scenario with $\numLoops = 120$, we observe that VoI scheduler succeeds at reducing average error per loop and fails at ensuring information freshness at the controller. As long as downlink is the bottleneck, i.e., $\numOfULResources \geq \numOfDLResources$, AoI scheduler does not perceive any performance gain by an increase of $\numOfULResources$. That follows from the definition of AoI scheduler in Section \ref{subsec:AoI}. Since age shows a deterministic behavior, no additional resources are used on the UL unless the packets are going to be forwarded in the next transmission opportunity. However, VoI benefits from every additional UL resource since $\bs$ is able to reduce the uncertainty of a sub-system at the $\bs$. Thus, it gets the chance to prefer some \textit{more valuable} packets over the ones carrying lower valued information by examining the packet content. As a result, we observe an ongoing but converging decrease in $\iae$ as we move from $\numOfULResources=1$ towards $\numOfULResources = 9$. Similarly, by virtue of additional UL resources, the loops which never get the opportunity before, find the chance to transmit. Hence, the average age $\avgAge$ gets a finite value again for $N=120$ and $\numOfULResources \in \{6, \, 9\}$. 

\begin{figure}[t]
	\includegraphics[width=\linewidth]{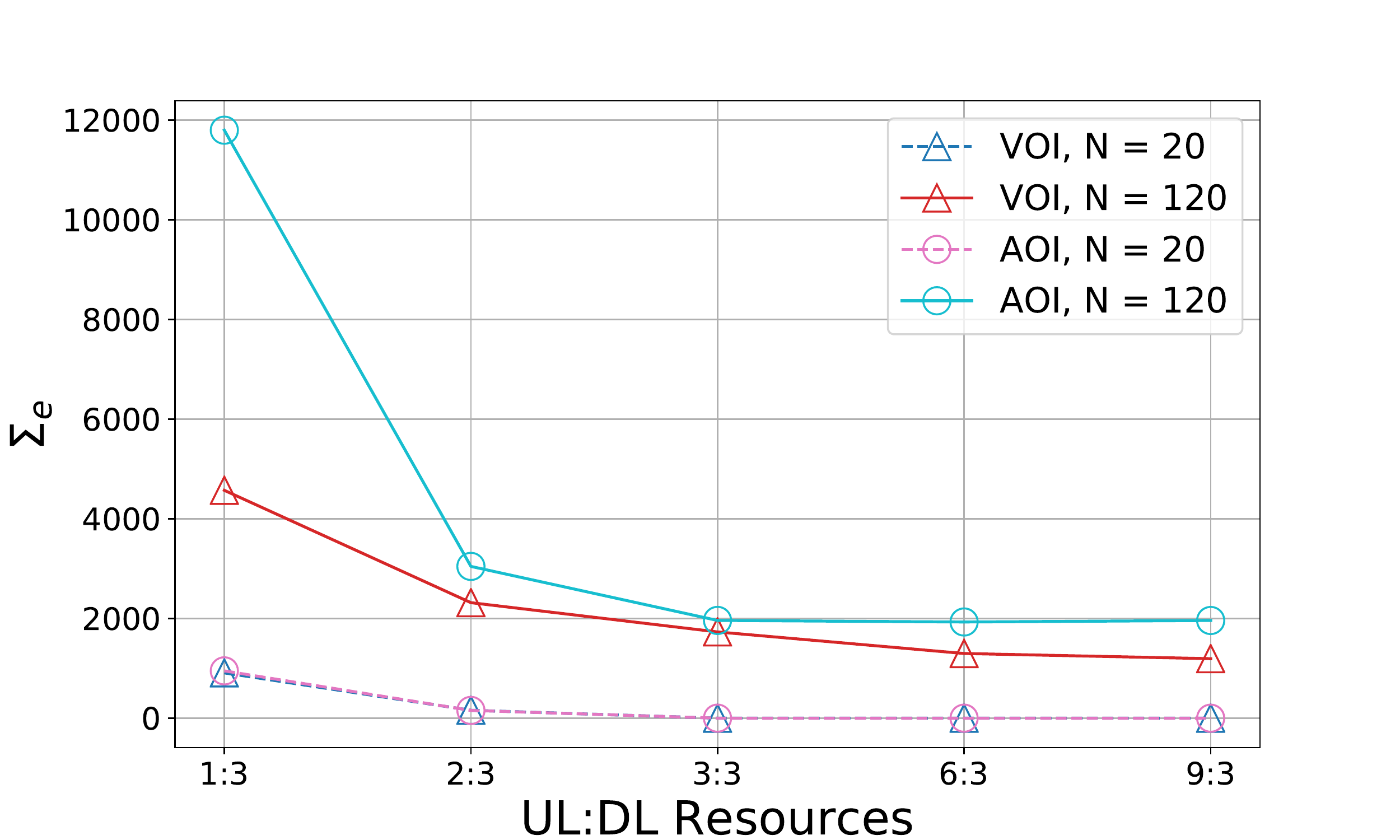}
	\caption{Sensitivity analysis of the average quadratic error norm to the UL/DL configuration, with the ratio $\frac{\numOfULResources}{\numOfULResources}$ on the $x$-axis. The number of DL resources is kept fixed $\numOfDLResources=3$, and the number of UL resources is varying $\numOfULResources\in\{1,\dots,9\}$. Left part of the plot ($\numOfULResources<3$) represents an UL bottleneck scenario, while the right part ($\numOfULResources>3$) represents a DL bottleneck scenario.}
	\label{fig:INIE_sens}
\end{figure}

%% file: relatedwork.tex
\section{Related Work}
\label{sec:relatedwork}
Cross-layer network design~\cite{cervin2008scheduling, liu2004wireless, park2018wireless, li2016wireless} has attracted researchers by virtue of providing higher quality-of-control to networked control applications. Control-aware MAC strategies have been proposed for contention-based access in~\cite{vilgelm2016adaptive, gatsis2016control}, and for contention free-access in~\cite{walsh2001scheduling, vilgelm2017control, vasconcelos2017optimal, mamduhi2015robust}. In \cite{walsh2001scheduling}, authors study the centralized scheduling problem with multiple control loops closed over a shared communication channel. They assume the injection of error reports into control traffic in wired industrial networks. Each sensor reports the estimation error to the scheduler. Free network resources are distributed among sub-systems starting from the ones with the maximum error. 
In \cite{vilgelm2017control} authors compare control-aware scheduling to control-unaware schedulers in a single-hop cellular networked control systems with varying channel qualities among loops. They show that, the proposed control-aware scheduler outperforms the control-unaware schedulers such as proportional fair and maximum-throughput with respect to quality-of-control. 
\cite{vasconcelos2017optimal} studies \textit{one-shot} joint scheduling and estimation problem under resource constraints. In their work, they consider a network shared by multiple sensor and estimator pairs. Given the probabilistic distributions of individual states, centralized scheduler chooses a single sensor-estimator pair to transmit. They show that it is globally optimal to choose the maximum quadratic norm as scheduling and mean-value estimation as the estimation strategy. As the name one-shot suggests, the work focuses only on a single transmission decision and does not consider application-dependent propagation of estimation error over multiple time-steps. \cite{mamduhi2015robust} considers a two-level scheduling problem, i.e., sensors drop their packet locally based on a predefined error threshold value and a centralized scheduler allocates resources probabilistically among the contenting control loops. The scheduler collects local error information from each control loop as in \cite{walsh2001scheduling} and calculates channel access grant probabilities based on the reported value.

The cross-layer design problem has been generalized by the introduction of the concept of the AoI~\cite{kaul2012real}. AoI has defined the notion of information freshness, uniform for all applications. Many recent works have taken on the problem of scheduling with AoI-based utility~\cite{kaul2012status, costa2016on, hsu2017age, sun2017update}. Most relevant for cross-layer design, Kosta~\textit{et al.}~\cite{kosta2017age} introduce the term value-of-information (VoI), and study the case with its non-linear behavior. In this work, we go one step further and define the VoI as a functional of age and system dynamics of individual control applications. For the joint DL/UL scheduling, typical for cellular network scenarios, we compare VoI and AoI scheduling approaches with respect to the resulting NCS performance.

%% file: conclusions.tex
\section{Conclusions}
\label{sec:conclusions}
Age-of-Information is a newly introduced measure to capture information freshness from the application layer perspective. It has been used for data scheduling in multi-user scenarios. In the context of two-hop networked control systems, we were able to show that AoI alone does not capture the requirements of networked control loops. In addition to age, the evolution of uncertainty in the system over time is highly dependent of the application. We were able to formulate the estimation error as a function of the AoI and application specific system parameters. We have shown that using the VoI as scheduling metric leads to reduced estimation error in the stochastic process than providing regular updates to each sub-system. 

%% file: ack.tex
\begin{acks}
	This work has been carried out with the support of the German Research
	Foundation (DFG) grant KE1863/5-1 within the Priority Program SPP 1914 ``Cyber-Physical Networking''.
\end{acks}

%% file: appendix.tex
\section{Proof of Lemma \ref{lem:Lemma1}}
\label{app:lemma1}
\begin{proof}[Proof of Lemma \ref{lem:Lemma1}]\label{app:proof_lemma_1}
	Given $\Delta_i[\ki] > 0$ as in (\ref{eq:Cestimation}), it holds that:
	\begin{align*}
	\hat{x}_i[\ki] & = \E\left[x_i[k]~|~\mathcal{I}_i[\ki]\right] \nonumber\\
	& = \E\left[A_i x_i[\ki-1] + B_i u_i[\ki-1] + w_i[\ki-1]~|~\mathcal{I}_i[\ki]\right] \nonumber \\
	& = \E\big[A_i (A_i x_i[\ki-2] + B_i u_i[\ki-2] + w_i[\ki-2])\nonumber\\
	&\quad  + B_i u_i[\ki-1] + w_i[\ki-1]~|~\mathcal{I}_i[\ki]\big]\nonumber\\
	& = \E \bigg[A_i^{\Delta_i[\ki]} z_i[s_i] + \sum_{q=1}^{\Delta_i[\ki]} A_i^{q-1} w_i[\ki-q] \nonumber\\
	& \quad + \sum_{q=1}^{\Delta_i[\ki]}A_i^{q-1} B_i u_i[\ki-q] ~\big|~\mathcal{I}_i[\ki]\bigg]\nonumber\\
		& = A_i^{\Delta_i[\ki]} z_i[s_i] + \sum_{q=1}^{\Delta_i[\ki]} A_i^{q-1} B_i u_i[k-q]\nonumber  	\\
	\end{align*} 
\end{proof}
\section{Proof of Lemma \ref{lem:Lemma2}}
\label{app:lemma2}
\begin{proof} [Proof of Lemma (\ref{lem:Lemma2})]
	\label{app:proof_lemma_2}
	Given $\Delta_i[\ki] > 0$:
	\begin{align*}
		&\errornorm =  \quad \E\left[\left(e_i[\ki]\right)^T e_i[\ki]\right] \\
		&= \quad \E\left[\left(\sum_{r=1}^{\Delta_i[\ki]}A_i^{r-1} w_i[\ki-r]\right)^T \sum_{r=1}^{\Delta_i[\ki]}A_i^{r-1} w_i[\ki-r]\right]\\
		&= \quad \E\left[\sum_{r=1}^{\Delta_i[\ki]} \left(w_i[\ki-r]\right)^T \left(A_i^{r-1}\right)^T \sum_{r=1}^{\Delta_i[\ki]} A_i^{r-1} w_i[\ki-r]\right]\\
		&\stackrel{\mathclap{\normalfont\mbox{(1)}}}{=} \quad \E\Big[\sum_{r=1}^{\Delta_i[\ki]} (w_i[\ki-r])^T (A_i^{r-1})^T A_i^{r-1} w_i[\ki-r]\Big] \\
		&\stackrel{\mathclap{\normalfont\mbox{(2)}}}{=} \quad \sum_{r=1}^{\Delta_i[\ki]} \text{tr}((A_i^{r-1})^T A_i^{r-1} W_i)\\
		 &= \quad \sum_{r=0}^{\Delta_i[\ki]-1} \text{tr}((A_i^{r})^T A_i^{r} W_i),
	\end{align*}\\ \qedhere \\
	where $W_i = \E\Big[w_i[\ki-r] (w_i[\ki-r])^T\Big]$ is the noise covariance matrix. In step (1) it was used that noise vectors are i.i.d. and hence uncorrelated and step (2) holds because expectation of a quadratic norm of a random vector $x$ with covariance matrix $C$ is $\E\left[x^TAx\right] = \left(\E[x]\right)^T A~\E[x] + \tr(AC)$.
\end{proof}